\documentclass[prx,aps,superscriptaddress,footinbib,twocolumn]{revtex4-2}

\usepackage{amsthm,amsmath,amssymb}
\usepackage{mathrsfs}

\newcommand{\ket}[1]{\vert{ #1 }\rangle}

\usepackage{amsthm}
\newtheorem{theorem}{Theorem}

\newtheorem{lemma}{Lemma}
\newtheorem{corollary}{Corollary}

\theoremstyle{definition}
\newtheorem{definition}{Definition}
\theoremstyle{remark}

\usepackage{comment}
\usepackage{amssymb}
\usepackage{graphicx}
\usepackage{epstopdf}
\newcommand{\figpath}{.}

\usepackage{algorithm}
\usepackage[noend]{algpseudocode}

\usepackage{natbib}

\usepackage[colorlinks=true,citecolor=blue,linkcolor=magenta]{hyperref}

\usepackage{xcolor}

\usepackage{comment}

\begin{document}

\title{Constant-Overhead Magic State Injection into qLDPC Codes with Error Independence Guarantees}

\author{Guo Zhang}
\thanks{These authors contributed equally to this work.}
\affiliation{Graduate School of China Academy of Engineering Physics, Beijing 100193, China}

\author{Yuanye Zhu}
\thanks{These authors contributed equally to this work.}
\affiliation{Center on Frontiers of Computing Studies, Peking University, Beijing 100871, China}
\affiliation{School of Computer Science, Peking University, Beijing 100871, China}

\author{Xiao Yuan}
\affiliation{Center on Frontiers of Computing Studies, Peking University, Beijing 100871, China}
\affiliation{School of Computer Science, Peking University, Beijing 100871, China}

\author{Ying Li}
\email{yli@gscaep.ac.cn}
\affiliation{Graduate School of China Academy of Engineering Physics, Beijing 100193, China}

\begin{abstract}
Magic states are essential yet resource-intensive components for realizing universal fault-tolerant quantum computation. Preparing magic states within emerging quantum low-density parity-check (qLDPC) codes poses additional challenges, due to the complex encoding structures.
Here, we introduce a generic and scalable method for magic state injection into arbitrarily selected logical qubits encoded using qLDPC codes. Our approach, based on parallelized code surgery, supports the injection from either physical qubits or low-distance logical qubits. For qLDPC code families with asymptotically constant encoding rates, the method achieves injection into $\Theta(k)$ logical qubits---where $k$ denotes the logical qubit number of the code---with only constant qubit overhead and a time complexity of $\tilde{O}(d^2)$, where $d$ is the code distance.
A central contribution of this work is a rigorous proof that errors affecting the injected magic states remain independent throughout the procedure. This independence ensures the resilience of logical qubits against interactions with noisy ancillae and preserves the presumption of subsequent magic state distillation protocols. We further support our theoretical results with numerical validation through circuit-level simulations.
These findings advance the feasibility of scalable, fault-tolerant universal quantum computing using qLDPC codes, offering a pathway to significantly reduced qubit resource requirements in magic state injection.
\end{abstract}

\maketitle

Magic states are essential for implementing non-Clifford gates, playing a critical role in enabling quantum computational advantage over classical methods \cite{PhysRevA.70.052328,Bravyi2005,Chitambar2018QuantumRT}. Within quantum error correction, however, preparing magic states remains a significant and resource-intensive step toward achieving a universal set of logical operations. Although substantial progress has been made in magic state preparation for conventional quantum error-correcting codes, such as surface codes \cite{Eastin2009,Campbell2017,Gidney2024,Chen2025,2025arXiv250201743V}, efficiently preparing magic states in the emerging quantum low-density parity-check (qLDPC) codes continues to present notable challenges. Despite their advantage of significantly higher encoding rates compared to surface codes \cite{Gottesman2014,Breuckmann2021,Fowler2012,OGorman2017,Babbush2018}, qLDPC codes typically involve complex encoding structures, making logical operations---including magic state preparation---particularly difficult to realize efficiently.

Recent studies have demonstrated how code surgery methods \cite{Horsman2012,Litinski2019,Vuillot2019,Chamberland2022} can be adapted to quantum low-density parity-check (qLDPC) codes for performing logical measurements and implementing Clifford gates. Initially, these protocols required ancilla systems scaling as $O(d^2)$ for single-qubit logical operations, a cost comparable to that of conventional surface codes \cite{Cohen2022,Xu2024,Cowtan2024,Poirson2025}. Subsequent developments have significantly reduced this overhead to $\tilde{O}(d)$ by employing gauging measurements, provided the logical operators scale with a weight of $O(d)$ \cite{Williamson2024,Ide2024,Cross2024}. Additionally, several parallelization techniques have been introduced to enhance operational speed, such as optimized ancilla system designs, strategic code branching for logical operator decoupling, and integration with transversal gates \cite{Zhang2025,Xu2024a}. Currently, the most advanced protocol combines gauging measurements with brute-force branching, enabling the execution of $\Theta(k)$ Clifford gates with constant qubit overhead and achieving a time complexity of $\tilde{O}(d^2)$ \cite{Cowtan2025}. However, developing methods for injecting magic states at speeds comparable to Clifford gate implementations while maintaining constant qubit overhead remains an outstanding challenge \cite{Xu2024,Stein2024,He2025}.

\begin{figure}[t]
\centering
\includegraphics[width=\linewidth]{\figpath/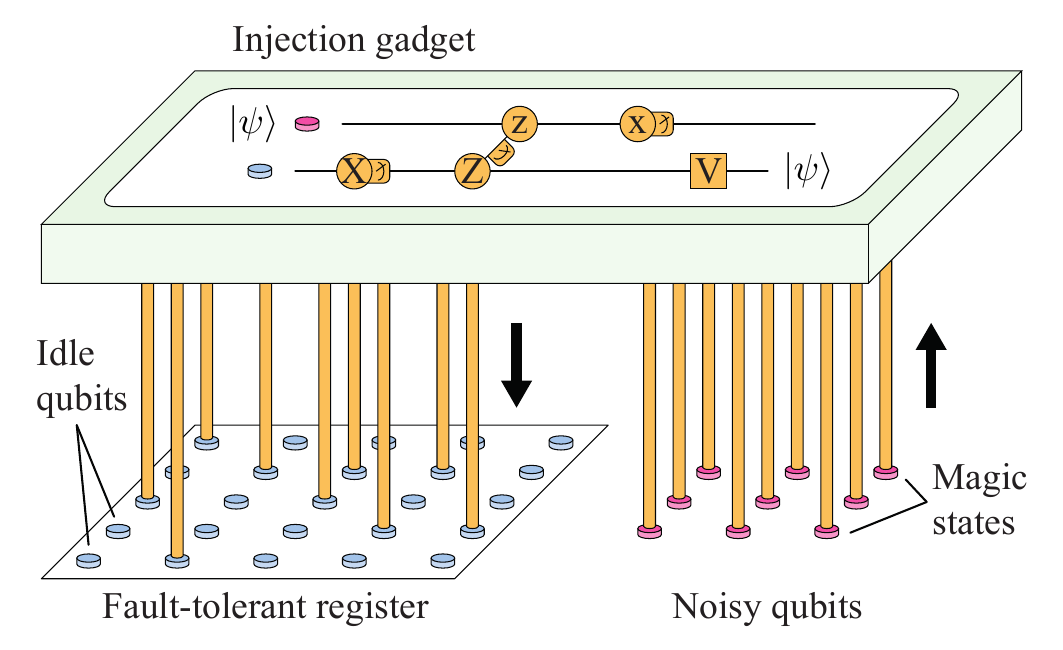}
\caption{
Magic state injection using joint logical measurements on fault-tolerant and noisy qubits. Magic states are prepared on noisy qubits and then transferred to arbitrarily chosen logical qubits in the register. Noisy qubits are either physical qubits or low-distance logical qubits encoded in such as low-distance surface codes and color codes. The register is encoded using a quantum low-density parity check code with a sufficiently large code distance to ensure fault tolerance. In the circuit, the gate $V = Z^{\mu_1+\mu_3}X^{\mu_2}$ depends on outcomes $\mu_1$, $\mu_2$, and $\mu_3$ of the $X$ measurement on the fault-tolerant logical qubit, $Zz$ joint measurement, and $x$ measurement on the noisy qubit, respectively. Each measurement outcome $\mu = 0,1$ corresponds to an operator eigenvalue $(-1)^{\mu}$. 
}
\label{fig:scheme}
\end{figure}

Here, we address this challenge by proposing a scheme to inject magic states into $\Theta(k)$ logical qubits encoded with qLDPC codes, achieving constant qubit overhead and a time complexity of $\tilde{O}(d^2)$. The scheme begins by preparing magic states on noisy qubits---either physical qubits or logical qubits encoded in low-distance independent code blocks~\cite{khaneja2005optimal,Li2014AMS,2024arXiv240303991I,2023arXiv230212292G,Chamberland2020VeryLO}---and then transfers them in parallel to a fault-tolerant quantum register via code surgery (see Fig.~\ref{fig:scheme}). The use of noisy qubits ensures constant overhead, while the parallel transfer process reduces the time cost.

To rigorously establish the efficacy of the scheme, we address two critical challenges. First, interactions with noisy qubits may temporarily weaken the protection of logical qubits in the register. During code surgery, the effective code distance is limited by that of the noisy qubits, which may compromise full-distance protection. Second, the parallel transfer of multiple magic states can potentially introduce correlated errors, which would impair the effectiveness of subsequent magic state distillation. Since distillation protocols~\cite{Bravyi2005,Bravyi2012,Campbell2017,Krishna2019,Wills2024} assume independent errors in the input states, such correlations can significantly degrade the output fidelity (see Appendix~\ref{app:distillation}).

In this work, we provide a rigorous proof that our scheme overcomes both challenges: it preserves the fault tolerance of logical qubits during state transfer and ensures that errors on the injected magic states remain statistically independent. As a result, the injected states can be reliably distilled~\cite{Bravyi2012,Krishna2019,Wills2024}, with recent protocols achieving near-constant spacetime overhead~\cite{Nguyen2024}. Moreover, our scheme is compatible with magic state cultivation, wherein noisy qubits host and evolve magic states before injection into the fault-tolerant register~\cite{Gidney2024,Chen2025,2025arXiv250201743V}.

In the remainder of this paper, we first present the implementation of the injection scheme and demonstrate that it achieves the desired resource and time efficiencies using gauging measurements and brute-force branching. We then develop a general theorem that characterizes logical errors arising during code surgery on qLDPC codes (see Theorem~\ref{the}) \cite{Cohen2022,Cross2024,Ide2024,Williamson2024}. In contrast to prior approaches focused on code distance alone, our theorem evaluates the weight of individual logical error configurations (see Definition~\ref{def:distance}) and accounts for both data and measurement errors. Using this framework, we analyze the logical errors introduced during magic state injection and resolve the two central challenges. Finally, we validate our theoretical results through numerical simulations based on a realistic circuit-level error model.

\vspace{0.2cm}

\noindent\textbf{Protocols and resource costs} --- 
The overall scheme is illustrated in Fig.~\ref{fig:scheme}. Throughout, we refer to the logical qubits within the register as fault-tolerant qubits. Suppose we aim to inject magic states into $q$ fault-tolerant qubits, arbitrarily selected from a total of $k$ such qubits in the register. The first step is to prepare magic states on $q$ noisy qubits, where each noisy qubit is independently encoded to ensure that errors on the magic states remain uncorrelated prior to injection. The register itself may be encoded as either a single code block or multiple blocks.

Without loss of generality, we assume the magic states are injected into the first $q$ fault-tolerant qubits; relabeling suffices for arbitrary target sets. Let $\{(X_j, Z_j) : j \in I_k\}$ and $\{(x_j, z_j) : j \in I_q\}$ denote the logical Pauli operators of the fault-tolerant and noisy qubits, respectively, where $I_m = {1, 2, \dots, m}$ is the index set. The injection is implemented via a sequence of logical measurements on the operator set $\{X_j, Z_j z_j, x_j : j \in I_q\}$. These measurements are realized through code surgery \cite{Cohen2022,Zhang2025,Williamson2024}, as will be elaborated later.

Assume the register is encoded using an $[[n, k, d]]$ code from a family with a constant encoding rate \cite{Breuckmann2021,Panteleev2022,Leverrier2022}, i.e., $n = \Theta(k)$, and each noisy qubit is encoded using an $[[n_{\mathrm{noi}}, 1, d_{\mathrm{noi}}]]$ code. To minimize overhead, the noisy qubits may be encoded using low-distance codes, such that $d_{\mathrm{noi}} = O(\mathrm{Polylog}(d))$, which leads to a qubit cost of $n_{\mathrm{noi}} = O(\mathrm{Polylog}(d))$. The total encoding cost thus becomes $n + q n_{\mathrm{noi}} = O(k + q \mathrm{Polylog}(d))$. By contrast, if one were to use full-distance surface-code logical qubits for injection, each block would require $n_{\mathrm{noi}} = \Theta(d^2)$ qubits, resulting in a total cost of $O(k + q d^2)$~\cite{Xu2024,Swaroop2024}.

In addition to encoding overhead, extra physical qubits are required to implement code surgery. Among available protocols, gauging measurement combined with brute-force branching is particularly resource-efficient \cite{Cowtan2025}. This strategy incurs a qubit overhead of $O(q d f)$, where $f = \mathrm{Polylog}(q) + \mathrm{Polylog}(d)$. To maintain constant qubit overhead, one can inject $q = \Theta(k / (d f))$ magic states per shot, completing the injection of $\Theta(k)$ magic states in $\Theta(d f)$ shots. Since code surgery involves $O(d)$ rounds of parity-check measurements per shot, the total time cost becomes $O(d^2 f)$.
Alternatively, using the devised sticking protocol for code surgery \cite{Zhang2025} allows the simultaneous injection of $\Theta(k)$ magic states in a single shot, reducing the time complexity to $O(d)$ at the expense of increasing qubit overhead by a factor of $d$.

Finally, we discuss the choice of $d_{\mathrm{noi}}$. Since the injection process requires $O(d)$ rounds of measurements, the noisy qubits must be sufficiently protected to prevent excessive error accumulation. To ensure that the per-round error rate remains $O(1/d)$, it suffices to choose $d_{\mathrm{noi}} = O(\mathrm{Polylog}(d))$. Importantly, for moderate values of $d$, our scheme remains applicable even when noisy qubits are unencoded physical qubits (i.e., $d_{\mathrm{noi}} = 1$), further reducing hardware requirements. 

\vspace{0.2cm}

\noindent\textbf{Theoretical tools} ---
We now develop the theoretical tools necessary to prove the efficacy of our scheme. 

We begin by presenting the formalism underlying code surgery protocols. Each code surgery operation simultaneously measures a set of logical operators. We focus on measurements of $Z$ logical operators on an $[[n, k, d]]$ CSS code block, and the conclusions can be applied to $X$ logical operators. Following formalism of CSS codes, we represent the $[[n, k, d]]$ code by the tuple $(H_X, H_Z, J_X, J_Z)$, where $H_X$ and $H_Z$ ($J_X$ and $J_Z$) are the check (generator) matrices representing the $X$ and $Z$ stabilizer (logical) operators, respectively. These matrices satisfy the following constraints: $H_X H_Z^\mathrm{T} = H_X J_Z^\mathrm{T} = H_Z J_X^\mathrm{T} = 0$ and $J_X J_Z^\mathrm{T} = E_k$. The set of $Z$ logical operators to be measured simultaneously is represented by a matrix $\alpha J_Z$, where $\alpha \in \mathbb{F}^{q \times k}$ is a full-rank matrix. Each row of $\alpha J_Z$ corresponds to a $Z$ logical operator, and $q$ independent operators are measured simultaneously. Similarly, the unmeasured $X$ and $Z$ logical operators are represented by $\alpha_\perp J_X$ and ${\alpha_\perp^\mathrm{r}}^\mathrm{T} J_Z$, respectively, where $\alpha_\perp \in \mathbb{F}^{(k-q) \times k}$ is a full-rank matrix satisfying $\alpha_\perp \alpha^\mathrm{T} = 0$. Here, the superscript ‘r’ denotes the right inverse. 

The logical measurement is carried out using a deformed code derived from the original code $(H_X, H_Z, J_X, J_Z)$ and the target operators $\alpha J_Z$ to be measured. Although there are various ways to construct the deformed code \cite{Cohen2022,Cowtan2024,Zhang2025,Cross2024,Swaroop2024,Ide2024,Williamson2024,Cowtan2025}, their check matrices take the homological-measurement form 
\begin{eqnarray}
\bar{H}_X = \left(\begin{array}{cc}
H_X & T \\
0 & H_M
\end{array}\right) \text{ and } \bar{H}_Z = \left(\begin{array}{cc}
H_Z & 0 \\
S & H_G^\mathrm{T}
\end{array}\right). \notag 
\end{eqnarray}
To realize the desired simultaneous measurement, the matrices $S$, $T$, $H_G$, and $H_M$ are chosen to satisfy the following conditions: i) $H_X S^\mathrm{T} = T H_G$; ii) $H_M H_G = 0$; iii)  there exists a matrix $R$ such that $\alpha J_Z R S = \alpha J_Z$ and $H_G (\alpha J_Z R)^\mathrm{T} = 0$; and iv) there exists a matrix $\beta$ such that $\alpha_\perp J_X S^\mathrm{T} = \beta H_G$. To ensure fault tolerance, the pair $(H_G, H_M)$ is chosen to be $(d_R, S)$-bounded for some positive integer $d_R$. Additional properties and a detailed analysis of the deformed code are provided in Appendix~\ref{app:deformed}. 


\begin{definition}
{\bf $(d_R, S)$-bounded pair.} Let $H_G \in \mathbb{F}_2^{r_G \times n_G}$ and $H_M \in \mathbb{F}_2^{r_M \times r_G}$ be two matrices satisfying $H_M H_G = 0$. Given a matrix $S \in \mathbb{F}^{n \times n_G}$ and a positive integer $d_R$, the pair $(H_G, H_M)$ is said be $(d_R, S)$-bounded if and only if for every vector $v \in \mathbb{F}^{r_G}$ satisfying $H_M v^\mathrm{T} = 0$ and $\vert v \vert < d_R$, there exists a vector $u_\star \in \mathbb{F}^{n_G}$ such that $v^\mathrm{T} = H_G u_\star^\mathrm{T}$ and $\vert u_\star S \vert \leq \vert v \vert$. 
\label{def:bound}
\end{definition}

The $(d_R, S)$-bound condition generalizes the Cheeger-constant condition introduced in Refs.~\cite{Cross2024,Williamson2024,Ide2024} in two respects: it applies to ancilla systems constructed via the hypergraph product that can measure multiple logical operators, and it extends to cases where the $S$ matrix is not a projection. We note that gauging measurement, even when combined with brute-force branching, yields $d_R = \infty$, while the original protocol and the devised sticking approach can achieve any desired $d_R$ using $O(nd_R)$ physical qubits. 

Using the deformed code, we implement the logical measurement as follows. In addition to the original code block, code surgery employs an ancilla system consisting of a separate set of qubits. The original code block and the ancilla system correspond to the first and second columns of the deformed-code check matrices ($\bar{H}_X$ and $\bar{H}_Z$), respectively. Given an input state within the logical subspace of the original code $(H_X, H_Z)$, we initialize all ancilla-system qubits in the $\ket{+}$ state, perform $d_T$ rounds of parity-check measurements of the deformed code, and then measure the ancilla-system qubits in the $X$ basis. Afterward, we resume parity-check measurements of $(H_X, H_Z)$. Eigenvalues of the measured $Z$ logical operators $\alpha J_Z$ are extracted from the deformed-code parity-check measurements. See Appendix~\ref{app:measurement} for details. 

\begin{figure}[t]
\centering
\includegraphics[width=\linewidth]{\figpath/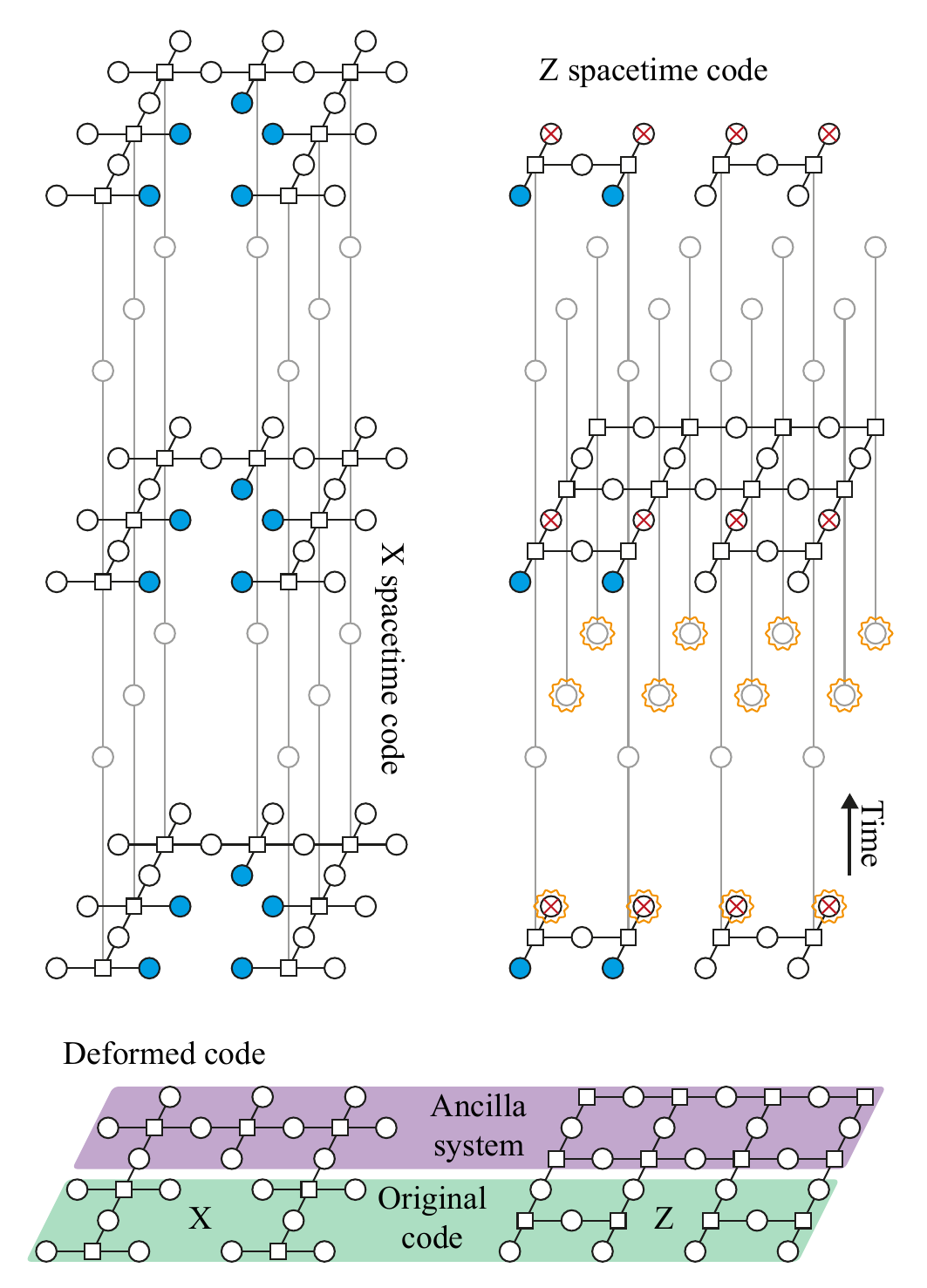}
\caption{
Tanner graphs of the deformed code and spacetime codes. We illustrate the measurement of the $Z \otimes Z$ operator on two surface-code logical qubits as an example. The surface code has a distance of two, and we set $d_R = d_T = 2$ in the measurement protocol. Circles represent bits, while squares represent check nodes. In spacetime codes, bits on horizontal edges correspond to data qubits in spacetime, and bits on vertical edges correspond to measurement outcomes of stabilizer operators. Blue-filled circles indicate the propagation of unmeasured logical operators $X \otimes X$ and $Z \otimes I$, whereas circles marked with red crosses represent the propagation of the measured logical operator $Z \otimes Z$. Circles enclosed by zigzagged outlines denote the extraction of the eigenvalue. 
}
\label{fig:spacetime}
\end{figure}

To facilitate a detailed analysis of logical errors in the quantum circuit of code surgery, we analyze the spacetime propagation of stabilizer and logical operators by introducing six characteristic matrices that encode their behavior. The propagation of each $X$ or $Z$ Pauli operator through spacetime can be represented by a vector. Each entry in the vector corresponds to either an $X$ ($Z$) operator acting on a qubit or the measurement outcome of an $X$ ($Z$) stabilizer operator at a specific time step in the circuit; see Fig.\ref{fig:spacetime}. We focus on six matrices formed from such vectors. The propagation of $X$ and $Z$ stabilizer operators gives rise to error checks throughout the circuit, captured by the spacetime check matrices $H^{st}_X$ and $H^{st}_Z$, respectively~\cite{Cross2024,Williamson2024}. The propagation of unmeasured $X$ and $Z$ logical operators, as well as measured $Z$ logical operators from the input state to the output state, is described by the generator matrices $J^{st}_X$, $J^{st}_Z$, and $J^{st}_{mz}$, respectively. Finally, the propagation of measured logical $Z$ operators from the input state to a set of stabilizer outcomes (the extraction of logical operator eigenvalues) is described by the generator matrix $J^{st}_{oc}$. The explicit construction of these spacetime matrices is provided in Appendix~\ref{app:spacetime}. 

Let $e_Z$ and $e_X$ be vectors representing $Z$ and $X$ errors at spacetime locations within the circuit. Errors are undetectable if $J^{st}_X e_Z^\mathrm{T} = 0$ and $J^{st}_Z e_X^\mathrm{T} = 0$. Logical operators remain unaffected by errors when $J^{st}_X e_Z^\mathrm{T} = J^{st}_Z e_X^\mathrm{T} = 0$ and $J^{st}_{mz} e_X^\mathrm{T} = 0$. The extracted eigenvalues of measured logical operators are correct when $J^{st}_{oc} e_X^\mathrm{T} = 0$. If $\psi = J^{st}_{oc} e_X^\mathrm{T}$ is not zero, then each one-valued entry of $\psi$ indicates that the corresponding eigenvalue is incorrect; similar interpretations apply to the other generator matrices. These spacetime matrices provide a framework for analyzing both logical errors within the circuit and errors in the extracted eigenvalues. 

Since the injection acts on a composite system comprising both high-distance and low-distance logical qubits, it is essential to assess the error protection of each logical operator and extracted eigenvalue individually. To this end, we introduce the following distance measure. 

\begin{definition}
{\bf Error-wise distance.} Given the check matrix $H \in \mathbb{F}_2^{r \times n}$, a generator matrix $J \in \mathbb{F}_2^{k \times n}$, and any vector $\psi \in \mathbb{F}_2^k$, the error-wise distance is defined as the minimum weight of an undetectable physical error $e$ that results in a logical error $\psi$. Specifically, 
\begin{eqnarray}
d(H, J, \psi) &=& \min\{\vert e \vert : e \in \mathrm{ker} H, J e^\mathrm{T} = \psi^\mathrm{T}\},
\end{eqnarray}
if there exists at least one $e \in \mathrm{ker} H$ satisfying $J e^\mathrm{T} = \psi^\mathrm{T}$. If no such $e$ exists, we define $d(H,J,\psi) = \infty$. 
\label{def:distance}
\end{definition}

\begin{lemma}
For all matrices $\phi \in \mathbb{F}_2^{p \times k}$ and vectors $\psi \in \mathbb{F}_2^p$, the error-wise distance satisfies 
\begin{eqnarray}
d(H, \phi J, \psi) &=& \min\{d(H, J, \psi_\star) : \psi_\star \in \mathbb{F}_2^k, \phi \psi_\star^\mathrm{T} = \psi^\mathrm{T}\}.~~~
\end{eqnarray}
\label{lem:distance}
\end{lemma}

We prove that the error-wise distance in code surgery is well lower bounded.

\begin{theorem}
In the code surgery on qLDPC codes, when both $d_R$ and $d_T$ are sufficiently large, the error-wise distances are lower-bounded by those of the original code. Specifically, the error-wise distances for errors on logical operators have the following lower bounds: 
\begin{eqnarray}
d(H^{st}_X, J^{st}_X, \psi) &\geq& \min\{d(H_X, \alpha_\perp J_X, \psi), d_R\},
\label{eq:disX}
\end{eqnarray} 
\begin{eqnarray}
d(H^{st}_Z, J^{st}_Z, \psi) &\geq& d(H_Z, {\alpha_\perp^\mathrm{r}}^\mathrm{T} J_Z, \psi),
\end{eqnarray}
and 
\begin{eqnarray}
d(H^{st}_Z, J^{st}_{mz}, \psi) &\geq& d(H_Z, \alpha J_Z, \psi);
\end{eqnarray}
and the error-wise distances of errors on extracted eigenvalues have the following lower bound: 
\begin{eqnarray}
d(H^{st}_Z, J^{st}_{oc}, \psi) &\geq& \min\{d(H_Z, \alpha J_Z, \psi),d_T\}.
\end{eqnarray}
\label{the}
\end{theorem}

\noindent See Appendices \ref{app:lemma} and \ref{app:theorem} for proofs of the lemma and theorem, respectively. 

\vspace{0.2cm}

\noindent\textbf{Error independence analysis} ---
We now use Theorem~\ref{the} to analyze logical errors in magic state injection. Our focus is on the joint measurement, as other operations neither compromise fault-tolerant qubits nor introduce correlated errors in magic states (see Fig.~\ref{fig:scheme}). Since the joint measurement is applied on the register and noisy qubits, the original code in code surgery becomes the one with check matrices $H_{X/Z} = H_{reg,X/Z} \oplus (E_q \otimes H_{noi,X/Z})$, where $H_{reg,X/Z}$ and $H_{noi,X/Z}$ are check matrices of the register and noisy qubit, respectively. Parameters of the original code become $[[n + qn_{noi}, k + q, \min\{d, d_{noi}\}]]$. In the joint measurement, we measure operators $\{Z_j z_j : j \in I_q\}$, which correspond to the generator matrices $J^{st}_{mz}$ and $J^{st}_{oc}$. 

There are two sets of unmeasured qubits. First, for each pair of qubits $(X_j, Z_j)$ and $(x_j, z_j)$ involved in the joint measurement, measuring $Z_j z_j$ leaves an effective unmeasured qubit with logical operators $(X_j x_j, Z_j)$. Without loss of generality, we assume that these unmeasured operators, $\{(X_j x_j, Z_j) : j \in I_q\}$, correspond to the first $q$ rows of the generator matrices $J^{st}_X$ and $J^{st}_Z$. Second, the logical operators of idle qubits are given by $\{(X_{q+j}, Z_{q+j}) : j \in I_{k-q}\}$, which correspond to the remaining $k-q$ rows of the generator matrices $J^{st}_X$ and $J^{st}_Z$. 

\begin{corollary}
Idle qubits are fault-tolerant in the joint measurement when $d_R\geq d$. 
\end{corollary}

The idle-qubit logical operator $X_{q+j}$ is represented by the $(q+j)$-th row of $\alpha_\perp J_X$, denoted as $(\alpha_\perp J_X)_{q+j,\bullet}$. The propagation of this operator is represented by the corresponding row in the generator matrix $J^{st}_X$, denoted as $(J^{st}_X)_{q+j,\bullet}$. Since the register has a code distance of $d$, we have $d(H_X, (\alpha_\perp J_X)_{q+j,\bullet}, 1)\geq d$. By applying Lemma~\ref{lem:distance} and using Eq.~(\ref{eq:disX}), we obtain the inequality $d(H^{st}_X, (J^{st}_X)_{q+j,\bullet}, 1)\geq d$. Similarly, for the idle-qubit logical operator $Z_{q+j}$, we have $d(H^{st}_Z, (J^{st}_Z)_{q+j,\bullet}, 1)\geq d$. These inequalities imply that at least $d$ single-qubit physical errors are required to cause a logical error on an idle qubit. 



Note that it is not necessary to demonstrate that the active qubits in the register are fault-tolerant during the joint measurement, as the injected magic states are inherently noisy. 

\begin{corollary}
Errors on injected magic states are effectively independent when $d_R,d_T\geq d$. 
\label{cor:independent}
\end{corollary}

We define that magic-state errors are effectively independent if and only if at least $\min\{pd_{noi}, d\}$ physical single-qubit errors are required to cause logical errors on $p$ magic states. Typically, the probability of a logical error decreases exponentially with the minimum number of physical errors required to cause the logical error. Consequently, this definition of effective independence suggests that the probability of $p$ magic-state errors decreases exponentially with $p$, until approaching the fault-tolerant level determined by the code distance $d$. 

Next, we prove the effective independence. Focusing on the joint measurement of $Z_j z_j$, there are four fundamental types of logical errors, from which all others are generated: Pauli errors $z_j$ (equivalent to $Z_j$), $x_j$, and $X_j$, as well as incorrect eigenvalue extraction. The error $z_j$ (and equivalently, $Z_j$) flips the logical operator $X_j x_j$ and can be analyzed using Eq.(\ref{eq:disX}). Since the noisy qubits have a code distance of $d_{noi}$, their error-wise distances have the lower bound $d(H_X, (\alpha_\perp J_X)_{1:q}, \psi) \geq \vert \psi \vert d_{noi}$, where the subscript $1:q$ denotes the first $q$ rows of the matrix. Applying Lemma~\ref{lem:distance}, we can find that $d(H^{st}_X, (J^{st}_X)_{1:q}, \psi) \geq \min\{\vert \psi \vert d_{noi}, d\}$. This result confirms that $Z$ logical errors satisfy the criteria for effective independence. 

Among the remaining three types of logical errors, the error $x_j$ is trivial, as it does not introduce any logical error on the magic state. The error $X_j$ flips the logical operator $Z_j$ and can be analyzed in the same way as idle-qubit $Z$ operators. Specifically, since $d(H^{st}_Z, (J^{st}_Z)_{j,\bullet}, 1)\geq d$, at least $d$ physical errors are required to induce an $X_j$ error. If fewer than $d$ physical errors occur, only incorrect eigenvalue extraction needs to be considered. Given that $d(H_Z, \alpha J_Z, \psi)\geq \vert \psi \vert d_{noi}$, it follows that $d(H^{st}_Z, J^{st}_{oc}, \psi) \geq \min\{\vert \psi \vert d_{noi}, d\}$. Thus, $X$ logical errors also satisfy the criterion for effective independence. 

\vspace{0.2cm}

\noindent\textbf{Numerical results} ---
To demonstrate the robustness of idle qubits and independence of magic state errors, we conduct numerical simulations of the joint measurement using a circuit-level error model. As an example, we consider a $[[90,8,10]]$ bivariate bicycle (BB) code \cite{Bravyi2024,Wang2024}. We simulate the joint measurement of operators $Z_1z_1$ and $Z_2z_2$, which act on two of the eight logical qubits in the BB code and two distance-2 surface-code logical qubits (serving as the noisy qubits). With this joint measurement, we can inject magic states into the first and second logical qubits in the register. The joint measurement is implemented via code surgery: We perform parity-check measurements of the deformed code for $d_T = 10$ rounds. In addition, before and after the deformed-code parity-check measurements, we perform parity-check measurements on both the BB code and surface codes for $d_T$ rounds, which provide reliable eigenvalues of their stabilizer operators. These eigenvalues are required for correcting errors during the deformed-code parity-check measurements. Logical error rates are evaluated using the Monte Carlo simulation; see Appendix~\ref{app:simulations} for a detailed description. The results are shown in Fig.~\ref{fig:plots}. 

For idle qubits, we observe an effective distance of $d_{cir} \approx 5$ [see Fig.~\ref{fig:plots}(a)], which is consistent with the expected behavior of a code with distance ten. We also report the error rates of idle qubits without performing code surgery, and find that the effective distance remains unchanged in both cases. This confirms that the fault tolerance of idle qubits is preserved during code surgery. 

For active qubits, we evaluate logical errors that could affect the injected magic states. As shown in Fig.~\ref{fig:plots}(b), the numerical results indicate that $\mathrm{Pr}(A \cap B) = \mathrm{Pr}(A)\mathrm{Pr}(B)$ and $\mathrm{Pr}(C \cap D) = \mathrm{Pr}(C)\mathrm{Pr}(D)$, demonstrating that errors affecting the magic states occur independently. 

\begin{figure}[t]
\centering
\includegraphics[width=\linewidth]{\figpath/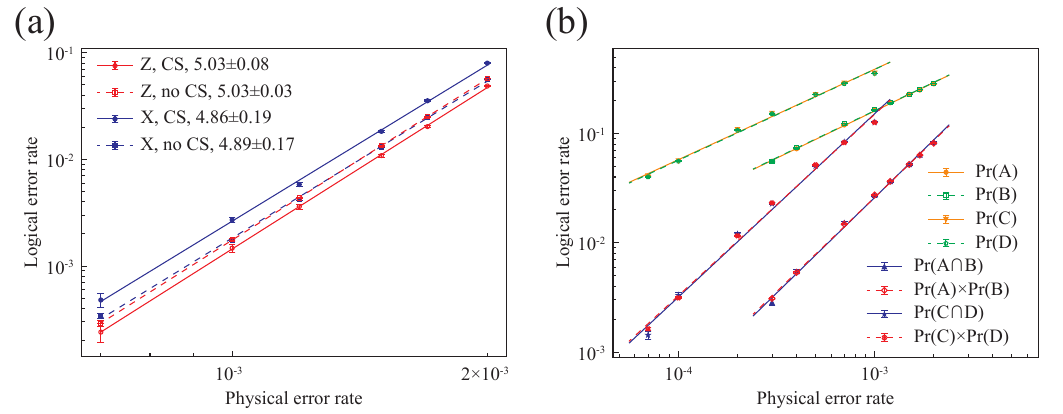}
\caption{
(a) Logical error rates of idle qubits with and without the code surgery (CS) performed. A logical error rate $p_L$ scales with the physical error rate per gate $p$ as $p_L = \alpha p^{d_{cir}}$, where the effective distance $d_{cir}$ is extracted by fitting and shown in the legend. 
(b) Logical error rates associated with faults on the injected magic states. Events $A$ and $B$ ($C$ and $D$) correspond to $Z$ ($X$) errors on the first and second injected magic states, respectively. 
}
\label{fig:plots}
\end{figure}

\vspace{0.2cm}

\noindent\textbf{Conclusions} ---
In this work, we proposed a magic-state injection scheme for qLDPC codes that minimizes time cost while maintaining constant qubit overhead. To support and justify this scheme, we developed a set of theoretical tools that enable fine-grained analysis of different error types arising during code surgery. These tools are broadly applicable and can be extended to protocols with non-uniform code distances; for example, assigning lower-distance codes to qubits that have limited influence on the final observable of a quantum circuit. 

By reducing the injection time cost to $\tilde{O}(d^2)$, our scheme enables universal fault-tolerant quantum computation at this speed using code surgery techniques, providing a practical path forward for efficient, low-overhead, fault-tolerant quantum computing. 

\begin{acknowledgments}
We would like to thank Hengyun Zhou and Qian Xu for providing numerical simulation programs and decoding packages, which provided important support for the smooth progress of this study.
This work is supported by the National Natural Science Foundation of China (Grant Nos. 12225507, 12088101) and NSAF (Grant No. U1930403). The source codes for the numerical simulation are available at \cite{code}. Yuanye Zhu and Xiao Yuan is supported by the Innovation Program for Quantum Science and Technology (Grant No.~2023ZD0300200), the National Natural Science Foundation of China Grant (No.~12175003 and No.~12361161602),  NSAF (Grant No.~U2330201). 
\end{acknowledgments}

\bibliography{references.bib}

\appendix

\begin{widetext}

\section{Impact of error correlations in magic state distillation}
\label{app:distillation}

One of the two key problems addressed in this work is the potential correlation between injected magic states. Using Corollary~\ref{cor:independent} and supporting numerical results, we demonstrate that the errors associated with the injected magic states remain statistically independent. This satisfies a critical requirement for subsequent magic state distillation, as all existing distillation protocols rely on the assumption that raw magic states exhibit independent errors. 

\begin{figure}[htbp]
\centering
\includegraphics[width=0.48\linewidth]{\figpath/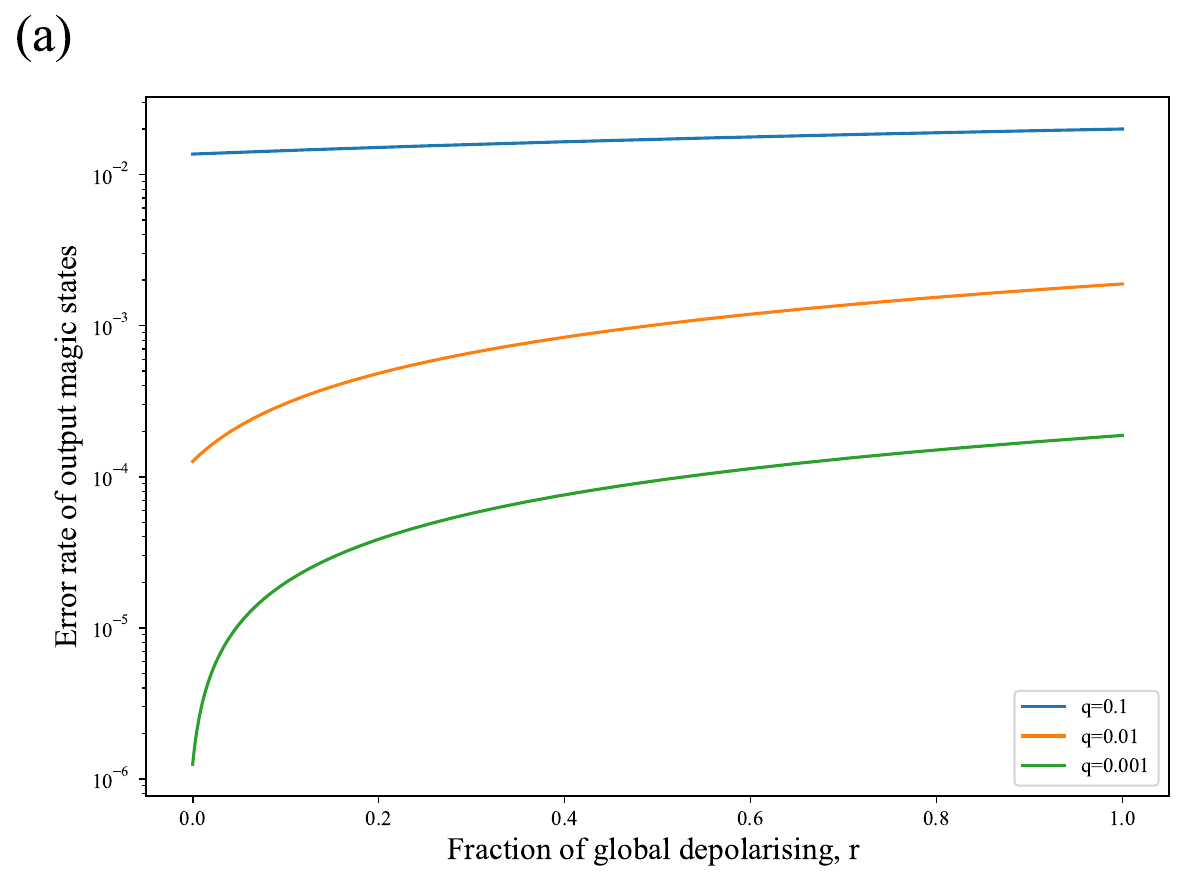}
\includegraphics[width=0.48\linewidth]{\figpath/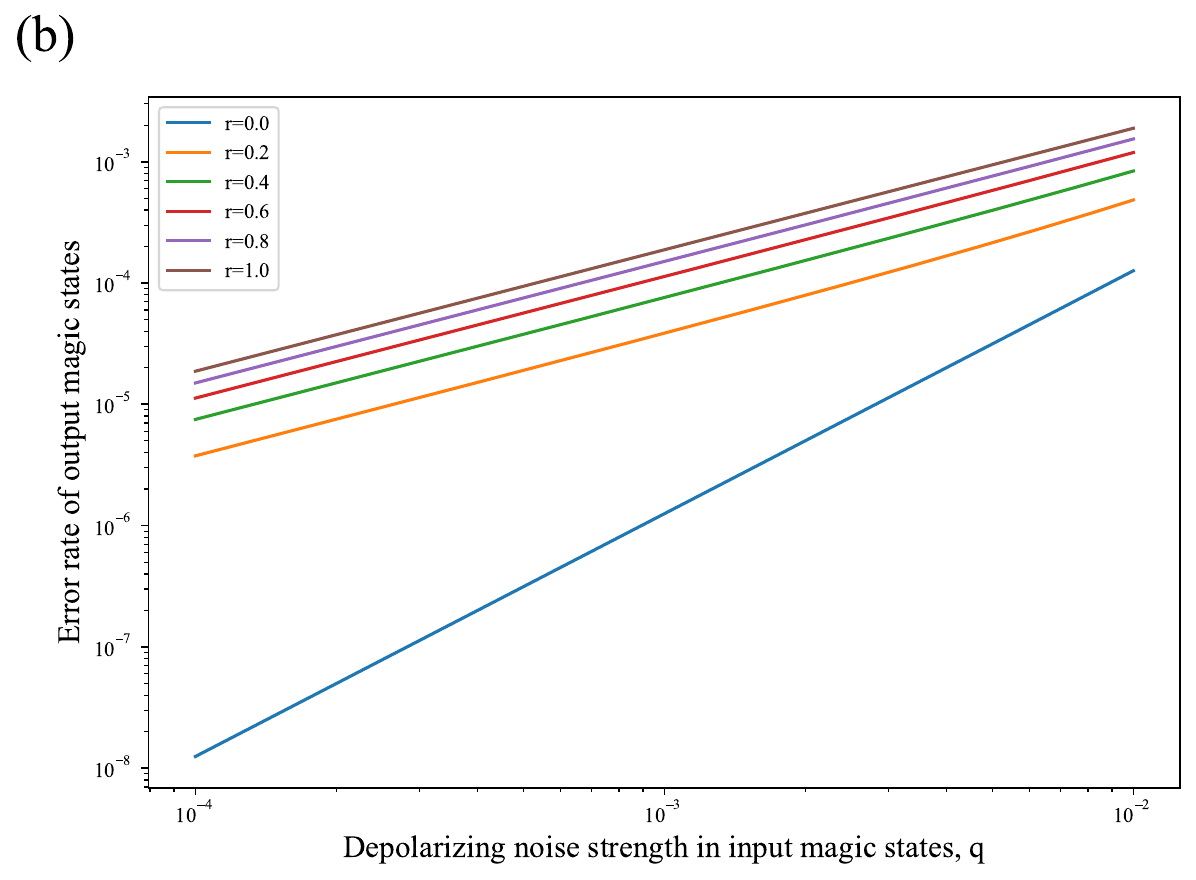}
\caption{
(a) Output error rate after magic state distillation as a function of the fraction of global depolarizing noise. 
(b) Distillation efficiency under uncorrelated and correlated noises. 
}
\label{fig:msd}
\end{figure}

Within the broader framework of magic-state-based fault-tolerant quantum computation, the final computational fidelity depends on both the magic state injection and the subsequent distillation processes. In our proposed scheme, noisy magic states are first injected into a fault-tolerant register, after which distillation is performed using fault-tolerant operations. Since the error rates of these logical operations are significantly lower than those of the injected magic states, they can be considered negligible for the purposes of our analysis. Consequently, our simulation focuses exclusively on the errors of injected magic states, with particular attention to the effects of correlated errors on the final output fidelity. To quantitatively evaluate the impact of such correlations on the effectiveness of distillation, we conducted the following simulations. 

We employ the standard 5-to-1 magic state distillation protocol~\cite{Bravyi2005,jochym2012study}, where five qubits are initialized in noisy magic states affected by both single-qubit depolarizing noise and global depolarizing noise. The global depolarizing noise introduces correlated errors. The input five-qubit state reads 
\begin{equation} 
\rho = (1-r) C_1^{\otimes 5} + r C_5,
\end{equation}
where $r$ denotes the fraction of global depolarizing noise, 
\begin{equation}
C_m = (1-q) A^{\otimes m} + q B_m,
\end{equation}
is a state of $m$ qubits prepared in the magic state and affected by $m$-qubit depolarizing noise, $A$ is the ideal magic state, $B_m$ is the maximally mixed state of $m$ qubits, and $q$ is the depolarizing rate. With the parameter $r$, we can model the interplay between independent and correlated noise among the five input qubits. After distillation, we calculate the error rate of the output magic state. Let $D$ be the output state. The error rate is 
\begin{equation} 
\text{Error rate} = 1 - \mathrm{Tr}(A D).
\end{equation}

The results are shown in Fig.~\ref{fig:msd}(a): the horizontal axis indicates the fraction of global depolarizing in the input states before distillation, the vertical axis shows the error rate of the output state after distillation, and different colors correspond to varying depolarizing noise levels. The data reveal that as the degree of correlated noise in the input increases, the output error rate after distillation increases significantly. 

To further illustrate the impact of correlated errors, Fig.~\ref{fig:msd}(b) shows results where the horizontal axis represents the depolarizing noise strength in the input magic states, the vertical axis indicates the output error rate after distillation, and different colors correspond to varying fractions of correlated errors. The results clearly show that as the level of correlated noise increases, the slope of the output error rate versus input depolarizing noise decreases from 2 to 1. This behavior indicates a qualitative shift in the performance of the distillation protocol: Once correlated errors are present, the protocol loses its characteristic $p$ to $p^2$ error suppression associated with independent noise, resulting in a reduced distillation efficiency. 

The numerical results clearly highlight the critical importance of error independence among injected magic states for achieving efficient distillation and high output fidelity. Our injection scheme proposed in this work ensures that errors across the injected states remain statistically independent. As a result, when combined with magic state distillation protocols, our scheme enables the preparation of high-fidelity magic states. 

\section{Error-wise distance --- Proof of Lemma \ref{lem:distance}}
\label{app:lemma}

If $\{e\in\mathrm{ker}H : \phi Je^\mathrm{T} = \psi^\mathrm{T}\} \neq \emptyset$, 
\begin{eqnarray}
d(H,\phi J,\psi) &=& \min\{\vert e\vert : e\in\mathrm{ker}H, \phi Je^\mathrm{T} = \psi^\mathrm{T}\} \notag \\
&=& \min\{\vert e\vert : e\in\mathrm{ker}H, \psi_\star^\mathrm{T} = Je^\mathrm{T}, \phi\psi_\star^\mathrm{T} = \psi^\mathrm{T}\} \notag \\
&=& \min\{d(H,J,\psi_\star) : \psi_\star\in\mathbb{F}_2^k, \phi\psi_\star^\mathrm{T} = \psi^\mathrm{T}\}.
\end{eqnarray}
If $\{e\in\mathrm{ker}H : \phi Je^\mathrm{T} = \psi^\mathrm{T}\} = \emptyset$, $\{e\in\mathrm{ker}H : Je^\mathrm{T} = \psi_\star^\mathrm{T}\} = \emptyset$ for all $\psi_\star\in\mathbb{F}_2^k$ with $\phi\psi_\star^\mathrm{T} = \psi^\mathrm{T}$; therefore, $d_P(H,J,\psi_\star) = \infty$. 

\section{Code surgery on qLDPC codes}
\label{app:measurement}

\subsection{Code surgery }

Let $H_X\in\mathbb{F}_2^{r_X\times n}$ and $H_Z\in\mathbb{F}_2^{r_Z\times n}$ be the $X$-operator and $Z$-operator check matrices of a CSS code, respectively: $H_XH_Z^\mathrm{T} = 0$. Let $J_X\in\mathbb{F}_2^{k\times n}$ and $J_Z\in\mathbb{F}_2^{k\times n}$ be the generator matrices of $X$ and $Z$ logical operators, respectively: $H_XJ_Z^\mathrm{T} = H_ZJ_X^\mathrm{T} = 0$ and $J_XJ_Z^\mathrm{T} = E_k$. 

Given a CSS code $(H_X,H_Z,J_X,J_Z)$, we can measure logical operators $\alpha J_Z$ for any full-rank matrix $\alpha\in\mathbb{F}_2^{q\times k}$ with $q\leq k$ using the following protocol: 
\begin{itemize}
\item[1.] Initialize all qubits in the ancilla system in the $X$ basis; 
\item[2.] Implement parity-check measurements of the following two check matrices for $d_T$ rounds, 
\begin{eqnarray}
\bar{H}_X &=& \left(\begin{array}{cc}
H_X & T \\
0 & H_M
\end{array}\right)
\end{eqnarray}
and 
\begin{eqnarray}
\bar{H}_Z &=& \left(\begin{array}{cc}
H_Z & 0 \\
S & H_G^\mathrm{T}
\end{array}\right),
\end{eqnarray}
where $S\in\mathbb{F}_2^{n_G\times n}$, $T\in\mathbb{F}_2^{r_X\times r_G}$, $H_G\in\mathbb{F}_2^{r_G\times n_G}$ and $H_M\in\mathbb{F}_2^{r_M\times r_G}$ satisfy conditions i) $H_XS^\mathrm{T} = TH_G$, ii) $H_MH_G = 0$, and iii) there exists $R\in\mathbb{F}_2^{n\times n_G}$ such that $\alpha J_ZRS = \alpha J_Z$ and $H_G(\alpha J_ZR)^\mathrm{T} = 0$; and iv) there exists $\beta\in\mathbb{F}_2^{(k-q)\times r_G}$ such that $\alpha_\perp J_XS^\mathrm{T} = \beta H_G$, where $\alpha_\perp\in\mathbb{F}_2^{(k-q)\times k}$ is a full-rank matrix with $\alpha_\perp\alpha^\mathrm{T} = 0$; 
\item[3.] Measure all qubits in the ancilla system in the $X$ basis. 
\end{itemize}
Notice that 
\begin{eqnarray}
\alpha J_ZR\left(\begin{array}{cc}
S & H_G^\mathrm{T}
\end{array}\right) = \left(\begin{array}{cc}
\alpha J_Z & 0
\end{array}\right).
\end{eqnarray}
Let $\mu\in\mathbb{F}_2^{r_X+n_G}$ be the outcomes in the first round (or any other rounds) of $\bar{H}_X$ parity-check measurements. The measurement outcomes of logical operators $\alpha J_Z$ are $\alpha J_ZR\gamma_2 \mu^\mathrm{T}$, where 
\begin{eqnarray}
\gamma_1 = \left(\begin{array}{cc}
E_{r_Z} & 0
\end{array}\right)
\end{eqnarray}
and 
\begin{eqnarray}
\gamma_2 = \left(\begin{array}{cc}
0 & E_{n_G}
\end{array}\right).
\end{eqnarray}

\subsection{Measurement sticker as a $(d_R, S)$-bounded pair}

In this section, we provide a proof of the $(d_R, S)$-bound in the {\it devised sticking} protocol~\cite{Zhang2025}. In devised sticking, a check matrix $H_G \in \mathbb{F}_2^{r_G \times n_G}$ is specifically designed to realize the desired simultaneous logical measurement on the original code $(H_X,H_Z,J_X,J_Z)$, and the corresponding code is referred to as the glue code. The glue code is said to be compatible with the original code if and only if there exist pasting matrices $S$ and $T$ such that $H_X S^\mathrm{T} = T H_G$. Note that we have used the same notations in the expressions of the deformed code, because these matrices satisfy the same conditions. We can always construct $H_G$ such that $\Vert S \Vert = 1$. Here, $\Vert \bullet \Vert$ is the matrix norm induced by Hamming weight, i.e., 
\begin{eqnarray}
\Vert S \Vert = \max\left\{ \frac{\vert uS \vert}{\vert u \vert} : u \neq 0 \right\}.
\end{eqnarray}
In the code surgery, we take deformed-code check matrices 
\begin{eqnarray}
\bar{H}_X &=& \left(\begin{array}{cc}
H_X & \tilde{T} \\
0 & \tilde{H}_M
\end{array}\right)
\end{eqnarray}
and 
\begin{eqnarray}
\bar{H}_Z &=& \left(\begin{array}{cc}
H_Z & 0 \\
\tilde{S} & \tilde{H}_G^\mathrm{T}
\end{array}\right),
\end{eqnarray}
where 
\begin{eqnarray}
\tilde{H}_G &=& \begin{pmatrix} \lambda_{d_R}^\mathrm{T} \otimes E_{n_G} & E_{d_R} \otimes H_G^\mathrm{T} \end{pmatrix}, \\
\tilde{H}_M &=& \begin{pmatrix} E_{d_R-1} \otimes H_G & \lambda_{d_R} \otimes E_{r_G} \end{pmatrix},
\end{eqnarray}
and 
\begin{eqnarray}
\tilde{S} = e_1 \otimes S,
\end{eqnarray}
where $e_1$ is the $d_R$-dimensional column vector with one in the first entry and zero elsewhere. 

\begin{lemma}
The pair $(\tilde{H}_G, \tilde{H}_M)$ is $(d_R, \tilde{S})$-bounded. 
\end{lemma}

\begin{proof}

Let $v = \left(\begin{array}{cc}a & b\end{array}\right)$. Then, the condition $\tilde{H}_M v^\mathrm{T} = 0$ becomes 
\begin{eqnarray}
(E_{d_R-1} \otimes H_G) a^\mathrm{T} &=& (\lambda_{d_R} \otimes E_{r_G}) b^\mathrm{T}.
\end{eqnarray}
Let's introduce matrices $A$, $B$, and $U$ such that 
\begin{eqnarray}
a^\mathrm{T} &=& \mathrm{vec}(A), \\
b^\mathrm{T} &=& \mathrm{vec}(B), \\
u_\star^\mathrm{T} &=& \mathrm{vec}(U).
\end{eqnarray}
We can rewrite the condition as 
\begin{eqnarray}
H_G A &=& B \lambda_{d_R}^\mathrm{T}.
\label{eq:AB}
\end{eqnarray}
Similarly, we can rewrite $v^\mathrm{T} = \tilde{H}_G u_\star^\mathrm{T}$ as 
\begin{eqnarray}
A &=& U \lambda_{d_R}^\mathrm{T}, \label{eq:A} \\
B &=& H_G U. \label{eq:B}
\end{eqnarray}

Because $\vert v \vert < d_R$, at least one column of $B$ is zero, i.e.,~it is in the form 
\begin{eqnarray}
B &=& \left(\begin{array}{ccc}
B_{\bullet,1:j-1} & 0 & B_{\bullet,j+1:d_R}
\end{array}\right).
\end{eqnarray}
Accordingly, we take the matrix $U$ in the form 
\begin{eqnarray}
U &=& \left(\begin{array}{ccc}
U_{\bullet,1:j-1} & 0 & U_{\bullet,j+1:d_R}
\end{array}\right).
\end{eqnarray}
Then, Eq.~(\ref{eq:A}) becomes 
\begin{eqnarray}
\left(\begin{array}{cc}
A_{\bullet,1:j-1} & A_{\bullet,j:d_R-1}
\end{array}\right) &=& \left(\begin{array}{ccc}
U_{\bullet,1:j-1} & 0 & U_{\bullet,j+1:d_R}
\end{array}\right) \left(\begin{array}{cc}
(\lambda_{d_R}^\mathrm{T})_{1:j-1,1:j-1} & 0 \\
(\lambda_{d_R}^\mathrm{T})_{j,1:j-1} & (\lambda_{d_R}^\mathrm{T})_{j,j:d_R-1} \\
0 & (\lambda_{d_R}^\mathrm{T})_{j+1:d_R,j:d_R-1}
\end{array}\right).
\end{eqnarray}
Because $(\lambda_{d_R}^\mathrm{T})_{1:j-1,1:j-1}$ and $(\lambda_{d_R}^\mathrm{T})_{j+1:d_R,j:d_R-1}$ are invertible, we have 
\begin{eqnarray}
U_{\bullet,1:j-1} &=& A_{\bullet,1:j-1} [(\lambda_{d_R}^\mathrm{T})_{1:j-1,1:j-1}]^{-1}, \\
U_{\bullet,j+1:d_R} &=& A_{\bullet,j:d_R-1} [(\lambda_{d_R}^\mathrm{T})_{j+1:d_R,j:d_R-1}]^{-1},
\end{eqnarray}
i.e., 
\begin{eqnarray}
U &=& \left(\begin{array}{ccc}
A_{\bullet,1:j-1} [(\lambda_{d_R}^\mathrm{T})_{1:j-1,1:j-1}]^{-1} & 0 & A_{\bullet,j:d_R-1} [(\lambda_{d_R}^\mathrm{T})_{j+1:d_R,j:d_R-1}]^{-1}
\end{array}\right) \\
&=& A \left(\begin{array}{ccc}
[(\lambda_{d_R}^\mathrm{T})_{1:j-1,1:j-1}]^{-1} & 0 & 0 \\
0 & 0 & [(\lambda_{d_R}^\mathrm{T})_{j+1:d_R,j:d_R-1}]^{-1}
\end{array}\right)
\end{eqnarray}
Using Eq.~(\ref{eq:AB}), we can verify that Eq.~(\ref{eq:B}) holds: 
\begin{eqnarray}
\left(\begin{array}{ccc}
B_{\bullet,1:j-1} & 0 & B_{\bullet,j+1:d_R}
\end{array}\right) &=& H_G A \left(\begin{array}{ccc}
[(\lambda_{d_R}^\mathrm{T})_{1:j-1,1:j-1}]^{-1} & 0 & 0 \\
0 & 0 & [(\lambda_{d_R}^\mathrm{T})_{j+1:d_R,j:d_R-1}]^{-1}
\end{array}\right) \\
&=& B \lambda_{d_R}^\mathrm{T} \left(\begin{array}{ccc}
[(\lambda_{d_R}^\mathrm{T})_{1:j-1,1:j-1}]^{-1} & 0 & 0 \\
0 & 0 & [(\lambda_{d_R}^\mathrm{T})_{j+1:d_R,j:d_R-1}]^{-1}
\end{array}\right) \\
&=& \left(\begin{array}{cc}
B_{\bullet,1:j-1} (\lambda_{d_R}^\mathrm{T})_{1:j-1,1:j-1} & B_{\bullet,j+1:d_R} (\lambda_{d_R}^\mathrm{T})_{j+1:d_R,j:d_R-1}
\end{array}\right) \\
&&\times \left(\begin{array}{ccc}
[(\lambda_{d_R}^\mathrm{T})_{1:j-1,1:j-1}]^{-1} & 0 & 0 \\
0 & 0 & [(\lambda_{d_R}^\mathrm{T})_{j+1:d_R,j:d_R-1}]^{-1}
\end{array}\right).
\end{eqnarray}
Therefore, there exists $u_\star$ such that $v^\mathrm{T} = \tilde{H}_G u_\star^\mathrm{T}$. 
Consider the explicit form 
\begin{eqnarray}
(\lambda_{d_R}^\mathrm{T})_{1:j-1,1:j-1} = \left(\begin{array}{cccccc}
1 & 0 & 0 & \cdots & 0 & 0 \\
1 & 1 & 0 & \cdots & 0 & 0 \\
0 & 1 & 1 & \cdots & 0 & 0 \\
\vdots & \vdots & \vdots & \ddots & \vdots & \vdots \\
0 & 0 & 0 & \cdots & 1 & 0 \\
0 & 0 & 0 & \cdots & 1 & 1
\end{array}\right)
\end{eqnarray}
we have the inverse matrix 
\begin{eqnarray}
[(\lambda_{d_R}^\mathrm{T})_{1:j-1,1:j-1}]^{-1} &=& \left(\begin{array}{cccccc}
1 & 0 & 0 & \cdots & 0 & 0 \\
1 & 1 & 0 & \cdots & 0 & 0 \\
1 & 1 & 1 & \cdots & 0 & 0 \\
\vdots & \vdots & \vdots & \ddots & \vdots & \vdots \\
1 & 1 & 1 & \cdots & 1 & 0 \\
1 & 1 & 1 & \cdots & 1 & 1
\end{array}\right)
\end{eqnarray}

Note that 
\begin{eqnarray}
\tilde{S}^\mathrm{T} u_\star^\mathrm{T} &=& \mathrm{vec}(S^\mathrm{T} U G_R^\mathrm{r}) = (S^\mathrm{T} U)_{\bullet,1} = S^\mathrm{T} U_{\bullet,1} = S^\mathrm{T} \sum_{i=1}^{j-1} A_{\bullet,i}.
\end{eqnarray}
Therefore, 
\begin{eqnarray}
\vert u_\star \tilde{S} \vert \leq \Vert S \Vert \times \vert a \vert \leq \Vert S \Vert \times \vert v \vert = \vert v \vert.
\end{eqnarray}
The lemma has been proved.
\end{proof}

\section{Error-wise distance of the deformed code}
\label{app:deformed}

\begin{lemma}
For the deformed code $(\bar{H}_X,\bar{H}_Z)$, the generator matrices of $X$ and $Z$ logical operators are 
\begin{eqnarray}
\bar{J}_X = \left(\begin{array}{cc}
\alpha_\perp J_X & \beta
\end{array}\right)
\end{eqnarray}
and 
\begin{eqnarray}
\bar{J}_Z = \left(\begin{array}{cc}
{\alpha_\perp^\mathrm{r}}^\mathrm{T} J_Z & 0
\end{array}\right),
\end{eqnarray}
respectively. If $(H_G,H_M)$ is $(d_R,S)$-bounded, the error-wise distances of $Z$ and $X$ logical errors are 
\begin{eqnarray}
d(\bar{H}_X,\bar{J}_X,\psi) &\geq& \min\{d(H_X,\alpha_\perp J_X,\psi),d_R\}
\end{eqnarray}
and 
\begin{eqnarray}
d(\bar{H}_Z,\bar{J}_Z,\psi) &\geq& d(H_Z,{\alpha_\perp^\mathrm{r}}^\mathrm{T} J_Z,\psi),
\end{eqnarray}
respectively. 
\label{lem:deformed}
\end{lemma}

\begin{proof}
The generator matrices satisfy conditions $\bar{H}_X\bar{J}_Z^\mathrm{T} = \bar{H}_Z\bar{J}_X^\mathrm{T} = 0$ and $\bar{J}_X\bar{J}_Z^\mathrm{T} = E_{k-q}$, i.e. they are valid generator matrices of the deformed code. 

Let $e = \left(\begin{array}{cc}u & v\end{array}\right)$ be a logical $X$ error satisfying $\bar{H}_Ze^\mathrm{T} = 0$ and $\bar{J}_Ze^\mathrm{T} = \psi^\mathrm{T}$, i.e. 
\begin{eqnarray}
H_Zu^\mathrm{T} &=& 0
\end{eqnarray}
and 
\begin{eqnarray}
{\alpha_\perp^\mathrm{r}}^\mathrm{T} J_Zu^\mathrm{T} = \psi^\mathrm{T}.
\end{eqnarray}
Therefore, 
\begin{eqnarray}
\vert e\vert \geq \vert u\vert \geq d(H_Z,{\alpha_\perp^\mathrm{r}}^\mathrm{T} J_Z,\psi).
\end{eqnarray}
The inequality of $d(\bar{H}_Z,\bar{J}_Z,\psi)$ has been proved. 

Let $e = \left(\begin{array}{cc}u & v\end{array}\right)$ be a logical $Z$ error satisfying $\bar{H}_Xe^\mathrm{T} = 0$ and $\bar{J}_Xe^\mathrm{T} = \psi^\mathrm{T}$, i.e. 
\begin{eqnarray}
H_Xu^\mathrm{T} &=& Tv^\mathrm{T}, \\
H_Mv^\mathrm{T} &=& 0,
\end{eqnarray}
and 
\begin{eqnarray}
\alpha_\perp J_Xu^\mathrm{T} + \beta v^\mathrm{T} = \psi^\mathrm{T}.
\end{eqnarray}
If $\vert v\vert < d_R$, there exists $u_\star\in\mathbb{F}^{n_G}$ such that $v^\mathrm{T} = H_Gu_\star^\mathrm{T}$ and $\vert u_\star S\vert\leq \vert v\vert$. Then, 
\begin{eqnarray}
H_Xu^\mathrm{T} &=& TH_Gu_\star^\mathrm{T}
\end{eqnarray}
and 
\begin{eqnarray}
\alpha_\perp J_Xu^\mathrm{T} + \beta H_Gu_\star^\mathrm{T} = \psi^\mathrm{T}.
\end{eqnarray}
Furthermore, we have 
\begin{eqnarray}
H_X(u+u_\star S)^\mathrm{T} &=& 0
\end{eqnarray}
and 
\begin{eqnarray}
\alpha_\perp J_X(u+u_\star S)^\mathrm{T} = \psi^\mathrm{T}.
\end{eqnarray}
Therefore, 
\begin{eqnarray}
\vert e\vert = \vert u\vert + \vert v\vert \geq \vert u\vert + \vert u_\star S\vert \geq \vert u+u_\star S\vert \geq d(H_X,\alpha_\perp J_X,\psi).
\end{eqnarray}
Otherwise, $\vert e\vert \geq \vert v\vert \geq d_R$. The inequality of $d(\bar{H}_X,\bar{J}_X,\psi)$ has been proved. 
\end{proof}

\section{Spacetime error correction codes}
\label{app:spacetime}

To construct spacetime error correction codes , we make the following assumptions. First, the logical measurement is performed on a state within the logical subspace of the original code $(H_X, H_Z)$, meaning that the state is a common eigenstate of the stabilizer operators with known eigenvalues. Second, after the logical measurement, the stabilizer operators of the original code $(H_X, H_Z)$ are measured without error. In practice, to ensure reliable eigenvalues, we may need to measure stabilizer operators for several rounds both before and after the logical measurement. 

In the correction of errors in spacetime, most checks are comparisons between parity-check measurement outcomes. Consider the $j$-th $X$ stabilizer operator, which is represented by $(\bar{H}_X)_{j,\bullet}$, and let $\mu_{j,t},\mu_{j,t+1} = 0,1$ be its measurement outcomes in the $t$-th and $(t+1)$-th parity-check rounds, respectively. Then, $\mu_{j,t} + \mu_{j,t+1}$ is a check, and its value should be zero without errors. The errors that can flip the check are the following: measurement errors flipping outcomes $\mu_{j,t}$ and $\mu_{j,t+1}$, and $Z$ errors on data qubits in the support of the stabilizer operator $(\bar{H}_X)_{j,\bullet}$ occurring between the two parity-check rounds. Therefore, such a check is represented by a vector in the form 
\begin{eqnarray}
\left(\begin{array}{ccc}
(\bar{H}_X)_{j,\bullet} & 1 & 1
\end{array}\right), \notag
\end{eqnarray}
where the two ones correspond to measurement errors. If we consider all stabilizer operators in the $t$-th and $(t+1)$-th parity-check rounds, the checks are represented by a matrix in the form 
\begin{eqnarray}
\left(\begin{array}{ccc}
\bar{H}_X & E_{r_X+r_M} & E_{r_X+r_M}
\end{array}\right), \notag
\end{eqnarray}
where the two identity matrices $E_{r_X+r_M}$ correspond to measurement errors in the two rounds, respectively. If we consider several rounds of parity-check measurements, the matrix becomes 
\begin{eqnarray}
\left(\begin{array}{ccccccc}
\bar{H}_X & 0 & \cdots & E_{r_X+r_M} & E_{r_X+r_M} & 0 & \cdots \\
0 & \bar{H}_X & \cdots & 0 & E_{r_X+r_M} & E_{r_X+r_M} & \cdots \\
\vdots & \vdots & \ddots & \vdots & \vdots & \vdots & \ddots \\
\end{array}\right),
\end{eqnarray}
where columns of $\bar{H}_X$ correspond to $Z$ errors on data qubits, and columns of $E_{r_X+r_M}$ correspond to measurement errors. It is similar for $Z$ stabilizer operators. 

Now, we consider the first and the last parity-check rounds. Note that all qubits in the ancilla system are initialized and measured in the $X$ basis. We can work out eigenvalues of deformed-code $X$ stabilizer operators from eigenvalues of original-code $X$ stabilizer operators and eigenvalues of ancilla-qubit $X$ operators; therefore, all the first-round and last-round $X$ parity-check measurements can be verified using eigenvalues of original-code $X$ stabilizer operators and eigenvalues of ancilla-qubit $X$ operators. However, we can only work out eigenvalues for a subset of deformed-code $Z$ stabilizer operators from eigenvalues of original-code $Z$ stabilizer operators; therefore, only a subset of the first-round and last-round $Z$ parity-check measurements can be verified using eigenvalues of original-code $Z$ stabilizer operators. Taking into account these checks, the check matrices for correcting $Z$ and $X$ errors in spacetime are 
\begin{eqnarray}
H^{st}_X &=& \left(\begin{array}{cccc}
\bar{H}_X & 0 & 0 & E_{d_T;1}\otimes E_{r_X+r_M} \\
0 & E_{d_T-1}\otimes \bar{H}_X & 0 & H_{d_T}\otimes E_{r_X+r_M} \\
0 & 0 & \bar{H}_X & E_{d_T;d_T}\otimes E_{r_X+r_M}
\end{array}\right)
\end{eqnarray}
and 
\begin{eqnarray}
H^{st}_Z &=& \left(\begin{array}{cccc}
H_Z & 0 & 0 & E_{d_T;1}\otimes \gamma_1 \\
0 & E_{d_T-1}\otimes \bar{H}_Z & 0 & H_{d_T}\otimes E_{r_Z+n_G} \\
0 & 0 & H_Z & E_{d_T;d_T}\otimes \gamma_1
\end{array}\right),
\end{eqnarray}
respectively. Here, $E_{d_T;j}\in\mathbb{F}^{d_T}$ is the vector that only the $j$-th entry takes the value of one, and $H_{d_T}$ is the chain-graph check matrix of the repetition code with a distance of $d_T$. In $H^{st}_Z$, the appearance of $H_Z$ represents that only stabilizer operators included in $H_Z$ can be verified in the first and last rounds. 

Similar to check matrices, the generator matrices that describe the propagation of unmeasured $X$ and $Z$ logical operators in spacetime are 
\begin{eqnarray}
J^{st}_X &=& \left(\begin{array}{cccc}
\bar{J}_X & 1_{d_T-1}\otimes \bar{J}_X & \bar{J}_X & 0
\end{array}\right)
\end{eqnarray}
and 
\begin{eqnarray}
J^{st}_Z &=& \left(\begin{array}{cccc}
{\alpha_\perp^\mathrm{r}}^\mathrm{T} J_Z & 1_{d_T-1}\otimes \bar{J}_Z & {\alpha_\perp^\mathrm{r}}^\mathrm{T} J_Z & 0
\end{array}\right),
\end{eqnarray}
respectively. Here, $1_{d_T-1}\in\mathbb{F}^{d_T-1}$ is the vector that all entries take the value of one. The generator matrix that describes the propagation of measured $Z$ logical operators in spacetime is 
\begin{eqnarray}
J^{st}_{mz} &=& \left(\begin{array}{cccc}
\alpha J_Z & 1_{d_T-1}\otimes (\alpha J_Z\eta) & \alpha J_Z & 0,
\end{array}\right)
\end{eqnarray}
where 
\begin{eqnarray}
\eta = \left(\begin{array}{cc}
E_n & 0
\end{array}\right).
\end{eqnarray}
The generator matrix that describes the measurement of logical operators is 
\begin{eqnarray}
J^{st}_{oc} &=& \left(\begin{array}{cccc}
\alpha J_Z & 0 & 0 & E_{d_T;1}\otimes (\alpha J_ZR\gamma_2)
\end{array}\right),
\end{eqnarray}
where the first entry corresponds to the case that logical operators $\alpha J_Z$ are flipped before the logical measurement is applied, and the second entry corresponds to the case that outcomes of the first-round parity-check measurements are flipped; both of them can effectively flip logical measurement outcomes. 

\section{Error-wise distance of spacetime codes --- Proof of Theorem~\ref{the}}
\label{app:theorem}

\subsection{Errors on unmeasured $X$ logical operators}

Let $e = \left(\begin{array}{cccc}u_0 & u & u_1 & v\end{array}\right)$ be a logical $Z$ error that satisfies $H^{st}_Xe^\mathrm{T} = 0$ and $J^{st}_Xe^\mathrm{T} = \psi^\mathrm{T}$, i.e. 
\begin{eqnarray}
\bar{H}_Xu_0^\mathrm{T} &=& (E_{d_T;1}\otimes E_{r_X+r_M})v^\mathrm{T}, \\
(E_{d_T-1}\otimes \bar{H}_X)u^\mathrm{T} &=& (H_{d_T}\otimes E_{r_X+r_M})v^\mathrm{T}, \\
\bar{H}_Xu_1^\mathrm{T} &=& (E_{d_T;d_T}\otimes E_{r_X+r_M})v^\mathrm{T}
\end{eqnarray}
and 
\begin{eqnarray}
\bar{J}_Xu_0^\mathrm{T} + (1_{d_T-1}\otimes \bar{J}_X)u^\mathrm{T} + \bar{J}_Xu_1^\mathrm{T} = \psi^\mathrm{T}.
\end{eqnarray}
Let $U$ and $V$ be matrices corresponding to $u$ and $v$ such that $u^\mathrm{T}$ and $v^\mathrm{T}$ are vectorization forms of $U$ and $V$. Then, the equations become 
\begin{eqnarray}
\bar{H}_Xu_0^\mathrm{T} &=& VE_{d_T;1}^\mathrm{T}, \\
\bar{H}_XU &=& VH_{d_T}^\mathrm{T}, \\
\bar{H}_Xu_1^\mathrm{T} &=& VE_{d_T;d_T}^\mathrm{T}
\end{eqnarray}
and 
\begin{eqnarray}
\bar{J}_Xu_0^\mathrm{T} + \bar{J}_XU1_{d_T-1}^\mathrm{T} + \bar{J}_Xu_1^\mathrm{T} = \psi^\mathrm{T}.
\end{eqnarray}

Using 
\begin{eqnarray}
E_{d_T;1}^\mathrm{T} + H_{d_T}^\mathrm{T}1_{d_T-1}^\mathrm{T} + E_{d_T;d_T}^\mathrm{T} = 0,
\label{eq:sum}
\end{eqnarray}
we have 
\begin{eqnarray}
\bar{H}_Xu_0^\mathrm{T} + \bar{H}_XU1_{d_T-1}^\mathrm{T} + \bar{H}_Xu_1^\mathrm{T} = 0.
\end{eqnarray}
We define the effective error 
\begin{eqnarray}
u_{eff}^\mathrm{T} = u_0^\mathrm{T} + U1_{d_T-1}^\mathrm{T} + u_1^\mathrm{T}.
\end{eqnarray}
Then, 
\begin{eqnarray}
\bar{H}_Xu_{eff}^\mathrm{T} = 0
\end{eqnarray}
and 
\begin{eqnarray}
\bar{J}_Xu_{eff}^\mathrm{T} = \psi^\mathrm{T}.
\end{eqnarray}
Therefore, 
\begin{eqnarray}
\vert e\vert \geq \vert u_{eff}\vert \geq d(\bar{H}_X,\bar{J}_X,\psi).
\end{eqnarray}
The lower bound of $d(\bar{H}_X,\bar{J}_X,\psi)$ is given in Lemma~\ref{lem:deformed}. The inequality of $d(H^{st}_X,J^{st}_X,\psi)$ has been proved. 

\subsection{Errors on unmeasured $Z$ logical operators}
\label{app:Xerror}

Let $e = \left(\begin{array}{cccc}u_0 & u & u_1 & v\end{array}\right)$ be a logical $X$ error that satisfies $H^{st}_Ze^\mathrm{T} = 0$ and $J^{st}_Ze^\mathrm{T} = \psi^\mathrm{T}$, i.e. 
\begin{eqnarray}
H_Zu_0^\mathrm{T} &=& (E_{d_T;1}\otimes \gamma_1)v^\mathrm{T}, \\
(E_{d_T-1}\otimes \bar{H}_Z)u^\mathrm{T} &=& (H_{d_T}\otimes E_{r_Z+n_G})v^\mathrm{T}, \\
H_Zu_1^\mathrm{T} &=& (E_{d_T;d_T}\otimes \gamma_1)v^\mathrm{T}
\end{eqnarray}
and 
\begin{eqnarray}
{\alpha_\perp^\mathrm{r}}^\mathrm{T} J_Zu_0^\mathrm{T} + (1_{d_T-1}\otimes \bar{J}_Z)u^\mathrm{T} + {\alpha_\perp^\mathrm{r}}^\mathrm{T} J_Zu_1^\mathrm{T} = \psi^\mathrm{T}.
\end{eqnarray}
Let $U$ and $V$ be matrices corresponding to $u$ and $v$ such that $u^\mathrm{T}$ and $v^\mathrm{T}$ are vectorization forms of $U$ and $V$. Then, the equations become 
\begin{eqnarray}
H_Zu_0^\mathrm{T} &=& \gamma_1 VE_{d_T;1}^\mathrm{T}, \label{eq:u0} \\
\bar{H}_ZU &=& VH_{d_T}^\mathrm{T}, \label{eq:U} \\
H_Zu_1^\mathrm{T} &=& \gamma_1 VE_{d_T;d_T}^\mathrm{T} \label{eq:u1}
\end{eqnarray}
and 
\begin{eqnarray}
{\alpha_\perp^\mathrm{r}}^\mathrm{T} J_Zu_0^\mathrm{T} + \bar{J}_ZU1_{d_T-1}^\mathrm{T} + {\alpha_\perp^\mathrm{r}}^\mathrm{T} J_Zu_1^\mathrm{T} = \psi^\mathrm{T}.
\end{eqnarray}

Using Eq.~(\ref{eq:sum}), we have 
\begin{eqnarray}
H_Zu_0^\mathrm{T} + \gamma_1 \bar{H}_ZU1_{d_T-1}^\mathrm{T} + H_Zu_1^\mathrm{T} &=& 0.
\end{eqnarray}
Notice that $\gamma_1 \bar{H}_Z = H_Z\eta$. We define the effective error 
\begin{eqnarray}
u_{eff}^\mathrm{T} &=& u_0^\mathrm{T} + \eta U1_{d_T-1}^\mathrm{T} + u_1^\mathrm{T}.
\end{eqnarray}
Then, 
\begin{eqnarray}
H_Z u_{eff}^\mathrm{T} &=& 0
\end{eqnarray}
and 
\begin{eqnarray}
{\alpha_\perp^\mathrm{r}}^\mathrm{T} J_Z u_{eff}^\mathrm{T} = \psi^\mathrm{T}.
\end{eqnarray}
Here, we have used that $\bar{J}_Z = {\alpha_\perp^\mathrm{r}}^\mathrm{T} J_Z \eta$. Therefore, 
\begin{eqnarray}
\vert e\vert \geq \vert u_{eff}\vert \geq d(H_Z,{\alpha_\perp^\mathrm{r}}^\mathrm{T} J_Z,\psi).
\end{eqnarray}
The inequality of $d(H^{st}_Z,J^{st}_Z,\psi)$ has been proved. 

\subsection{Errors on measured $Z$ logical operators}

Using the same notations as in Appendix~\ref{app:Xerror}, we consider a logical $X$ error that satisfies $H^{st}_Ze^\mathrm{T} = 0$ and $J^{st}_{mz}e^\mathrm{T} = \psi$. By following the same procedure as in Appendix~\ref{app:Xerror} and replacing $J^{st}_Z$ with $J^{st}_{mz}$, we can prove the inequality of $d(H^{st}_Z,J^{st}_{mz},\psi)$. 

\subsection{Errors on measurement outcomes}

Using the same notations as in Appendix~\ref{app:Xerror}, we consider a logical $X$ error that satisfies $H^{st}_Ze^\mathrm{T} = 0$ and $J^{st}_{oc}e^\mathrm{T} = \psi$. 
Then, Eqs. (\ref{eq:u0}), (\ref{eq:U}), and (\ref{eq:u1}) hold, and 
\begin{eqnarray}
\alpha J_Zu_0^\mathrm{T} + (\alpha J_ZR\gamma_2)VE_{d_T;1}^\mathrm{T} &=& \psi.
\label{eq:oc}
\end{eqnarray}

If $\vert e\vert \leq d_T$, at least one round of parity-check measurements is error-free. Suppose the $j$-th round is error-free, the $U$ and $V$ errors are in the form 
\begin{eqnarray}
U &=& \left(\begin{array}{cc}
U_{1:j-1} & U_{j:d_T-1}
\end{array}\right), \\
V &=& \left(\begin{array}{ccc}
V_{1:j-1} & 0 & V_{j+1:d_T}
\end{array}\right),
\end{eqnarray}
where $M_{a:b}$ denotes the columns from $a$ to $b$. Accordingly, Eqs. (\ref{eq:u0}), (\ref{eq:U}) and (\ref{eq:u1}) become 
\begin{eqnarray}
H_Zu_0^\mathrm{T} &=& \gamma_1 V_0E_{d_T;1}^\mathrm{T} \\
\bar{H}_ZU_{1:j-1} &=& V_0(H_{d_T}^\mathrm{T})_{1:j-1}, \\
\bar{H}_ZU_{j:d_T-1} &=& V_1(H_{d_T}^\mathrm{T})_{j:d_T-1}, \\
H_Zu_1^\mathrm{T} &=& \gamma_1 V_1E_{d_T;d_T}^\mathrm{T},
\end{eqnarray}
where 
\begin{eqnarray}
V_0 &=& \left(\begin{array}{ccc}
V_{1:j-1} & 0 & 0
\end{array}\right), \\
V_1 &=& \left(\begin{array}{ccc}
0 & 0 & V_{j+1:d_T}
\end{array}\right).
\end{eqnarray}
Using 
\begin{eqnarray}
V_0[E_{d_T;1}^\mathrm{T} + (H_{d_T}^\mathrm{T})_{1:j-1}1_{j-1}^\mathrm{T}] &=& 0, \\
V_1[(H_{d_T}^\mathrm{T})_{j:d_T-1}1_{d_T-j}^\mathrm{T} + E_{d_T;d_T}^\mathrm{T}] &=& 0, \\
\end{eqnarray}
we have 
\begin{eqnarray}
H_Zu_0^\mathrm{T} + \gamma_1 \bar{H}_ZU_{1:j-1}1_{j-1}^\mathrm{T} &=& 0, \\
\gamma_1 \bar{H}_ZU_{j:d_T-1}1_{d_T-j}^\mathrm{T} + H_Zu_1^\mathrm{T} &=& 0.
\end{eqnarray}

We define the effective error 
\begin{eqnarray}
u_{eff}^\mathrm{T} &=& u_0^\mathrm{T} + \eta U_{1:j-1}1_{j-1}^\mathrm{T}.
\end{eqnarray}
Then, 
\begin{eqnarray}
H_Zu_{eff}^\mathrm{T} &=& 0.
\end{eqnarray}
Using 
\begin{eqnarray}
VE_{d_T;1}^\mathrm{T} = V_0E_{d_T;1}^\mathrm{T} = V_0(H_{d_T}^\mathrm{T})_{1:j-1}1_{j-1}^\mathrm{T} = \bar{H}_ZU_{1:j-1} 1_{j-1}^\mathrm{T},
\end{eqnarray}
we rewrite Eq. (\ref{eq:oc}) as 
\begin{eqnarray}
\psi = \alpha J_Zu_0^\mathrm{T} + (\alpha J_ZR\gamma_2)\bar{H}_ZU_{1:j-1} 1_{j-1}^\mathrm{T} = \alpha J_Zu_0^\mathrm{T} + \alpha J_Z\eta U_{1:j-1} 1_{j-1}^\mathrm{T} = \alpha J_Zu_{eff}^\mathrm{T}.
\end{eqnarray}
Therefore, 
\begin{eqnarray}
\vert e\vert &\geq& \vert u_{eff}\vert \geq d(H_Z,\alpha J_Z,\psi).
\end{eqnarray}
The inequality of $d(H^{st}_Z,J^{st}_{oc},\psi)$ has been proved. 

\section{Numerical simulations}
\label{app:simulations}

\subsection{Numerical decoding simulations}

For the simulations in this work, we use Stim~\cite{Gidney2021} to simulate Clifford circuits. To decode error syndromes, we employ a two-stage quantum decoder: \textit{belief propagation plus ordered statistics decoding} (BP+OSD) \cite{Roffe2020,Roffe_LDPC_Python_tools_2022}. We implement the joint measurement through the \textit{devised sticking} protocol~\cite{Zhang2025}. 

We use the product coloration circuit for parity-check measurements, employing a single ancilla for each stabilizer generator~\cite{Xu2024}. The numerical simulations of the joint measurement begin by initializing logical qubits. We then simulate $d_T = 10$ parity-check measurement cycles of the deformed code. Additionally, before and after deformed-code parity-check measurements, $d_T = 10$ rounds of parity-check measurements are performed on both the register and the surface code. Finally, a transversal readout of the qubits is performed. 
 
In the parity-check measurement circuit, we consider a circuit-level depolarizing error model, characterized by a single error probability $p$. The specific error model per time step considered is as follows:
\begin{itemize}
    \item Qubits initialized in $\ket{0}$ or $\ket{+}$ are subject to depolarizing errors with probability $p$;
    \item Qubit pairs undergoing two-qubit operations are followed by two-qubit depolarizing errors, where each of the 15 non-identity two-qubit Pauli operators occurs with probability $\tfrac{p}{15}$;
    \item Idle qubits are subject to depolarizing errors with probability $p$;
    \item Measurement outcomes in the $X$ or $Z$ basis are subject to a bit-flip error with probability $\tfrac{2p}{3}$.
\end{itemize}

For details of the decoding in this work, we use a min-sum variant of the BP decoder~\cite{Emran2014} and an order-5 OSD with cost scaling (OSD-CS) decoder \cite{Panteleev2021}. For the min-sum BP decoder, we set the maximum number of iterations to 1000 (\texttt{max\_iter} = 1000) and the scaling factor to 0.9 (\texttt{ms\_scaling\_factor} = 0.9).

Logical error rates are evaluated using Monte Carlo simulations. We perform Monte Carlo sampling of the physical errors based on the circuit-level error model, then generate the syndrome and logical observable measurements corresponding to the sampled set of physical errors. Decoders then predict which errors occurred based on the parity-check measurements and correct the logical observable based on these predictions. A logical error occurs if the corrected logical observable measurement differs from the true noiseless value. The logical error rate is modeled as a binomial distribution, where each trial results in either a logical error or no error. Let $p_L$ denote the logical failure probability. The standard deviation of $p_L$ is given by $\sigma_{p_L}= \sqrt{(1-p_L)p_L/N}$, where $N$ is the number of samples.

\subsection{Logical errors in the joint measurement}

In the joint measurement, there are two types of errors that cause $X$ logical errors on injected magic states: errors that flip logical operators $Z_j$ and errors that flip measurement outcomes. Note that a flipped measurement outcome leads to an $X$ logical error on the corresponding fault-tolerant qubit due to the gate $V$. 

Let $e$ be a binary vector representing $X$ errors in spacetime, and let $EC(e)$ be a vector representing corresponding errors after error correction. We use $(Z_j)$ to denote the set of $e$ such that $EC(e)$ flips the logical operator $Z_j$, and we use $(oc_j)$ to denote the set of $e$ such that $EC(e)$ flips the measurement outcome of $Z_jz_j$; these errors correspond to the rows in $(J^{st}_Z)_{1:q}$ and $J^{st}_{oc}$, respectively. Then, the set of $e$ such that $EC(e)$ includes the logical error $X_j$ on the injected magic state is $[(Z_j)\cup(oc_j)] \setminus [(Z_j)\cap(oc_j)]$. The probability of the logical error $X_j$ on the injected magic state is $\mathrm{Pr}\{[(Z_j)\cup(oc_j)] \setminus [(Z_j)\cap(oc_j)]\}$. In the main text, $A=[(Z_1)\cup(oc_1)] \setminus [(Z_1)\cap(oc_1)]$ and $B=[(Z_2)\cup(oc_2)] \setminus [(Z_2)\cap(oc_2)]$. 

There is only one type of error that causes $Z$ logical errors on injected magic states: errors that flip logical operators $X_jx_j$. Let $e'$ be a vector representing $Z$ errors in spacetime, and let $EC(e')$ be a vector representing the corresponding errors after error correction. We use $(X_jx_j)$ to denote the set of $e'$ such that $EC(e')$ flips the logical operator $X_jx_j$. In the main text, $C = (X_1x_1)$ and $D = (X_2x_2)$. 

\end{widetext}

\end{document}